\documentclass[11pt]{article}
\usepackage{amsfonts}
\usepackage{mathrsfs}
\usepackage{amsthm}
\usepackage{amssymb}
\usepackage{amsmath}
\usepackage{enumerate}

\theoremstyle{plain}
\newtheorem{theorem}{Theorem}[section]
\newtheorem{proposition}[theorem]{Proposition}
\newtheorem{corollary}[theorem]{Corollary}
\newtheorem{lemma}[theorem]{Lemma}

\theoremstyle{definition}
\newtheorem{definition}{Definition}[section]

\theoremstyle{remark}
\newtheorem{remark}{\textbf{Remark}}[section]

\theoremstyle{example}

\numberwithin{equation}{section}

\title{Abstract Model of Continuous-Time Quantum Walk Based on Bernoulli Functionals and Perfect State Transfer}

\author{Ce Wang\\
    Yau Mathematical Sciences Center, Tsinghua University\\
    Beijing 100084, People's Republic of China}

\begin{document}

\maketitle

\noindent\textbf{Abstract.}\ \
In this paper, we present an abstract model of continuous-time quantum walk (CTQW) based on Bernoulli functionals
and show that the model has perfect state transfer (PST), among others.
Let $\mathfrak{h}$ be the space of square integrable complex-valued Bernoulli functionals, which is infinitely dimensional.
First, we construct on a given subspace $\mathfrak{h}_L \subset \mathfrak{h}$ a self-adjoint operator $\Delta_L$ via the canonical unitary involutions on $\mathfrak{h}$,
and by analyzing its spectral structure we find out all its eigenvalues.
Then, we introduce an abstract model of CTQW with $\mathfrak{h}_L$ as its state space, which is governed by the Schr\"{o}dinger equation with $\Delta_L$ as the Hamiltonian.
We define the time-average probability distribution of the model, obtain an explicit expression of the distribution,
and, especially, we find the distribution admits a symmetry property.
We also justify the model by offering a graph-theoretic interpretation to the operator $\Delta_L$ as well as to the model itself.
Finally, we prove that the model has PST at time $t=\frac{\pi}{2}$. Some other interesting results are also proven of the model.
\vskip 2mm

\vskip 2mm


\section{Introduction}

Continuous-time quantum walks (CTQWs) are quantum analogues of classical continuous-time random walks in probability theory,
and have found wide application in quantum computation, quantum communication as well as in modeling physical processes \cite{venegas}.

In their seminal paper \cite{farhi-gutmann}, Farhi and Gutmann initiated the study of CTQWs by introducing their model of CTQW on a graph.
Since then, many efforts have been made to understand evolution behavior of CTQWs and related aspects (see, e.g., \cite{venegas} and references therein).
In particular, Childs \cite{childs} has described precise correspondence between CTQWs and discrete-time quantum walks (DTQWs) on arbitrary graphs,
while Konno \cite{konno} has established a weak limit theorem for CTQWs on the integer lattice.

Perfect state transfer (PST) is a concept with great significance in quantum communication, which was originally introduced by Bose \cite{bose}
in the context of quantum spin chains and further studied by Christandl et al. \cite{christandl}.
Recently, there has been much interest in PST in the context of CTQWs.
Cheung and Godsil \cite{cheung-godsil} have addressed PST in CTQWs on cubelike graphs.
Alvir et al. \cite{alvir} have investigated PST in CTQWs relative to graph Laplacians.
Kempton, Lippner and Yau \cite{kempton} have considered graphs with a potential, and proven that CTQWs have PST on such graphs.
There are many other works on PST in the context of CTQWs (see, e.g., \cite{alvir, godsil, kendon,lin} and references therein).

Typically, a CTQW is associated with a graph $G=(V,E)$ in such a way that: the state space of the CTQW is the space $l^2(V)$
of square summable complex-valued functions on the vertex set $V$, and the evolution of the CTQW is governed by a Schr\"{o}dinger equation of the form
\begin{equation}\label{eq}
  \Phi_t = \mathrm{e}^{\mathrm{i}tH}\Phi_0,
\end{equation}
where $H$ is the Hamiltonian that respects the topology of the graph $G$ and $\Phi_t$ denotes the state of the CTQW at time $t$ with $\Phi_0$ being the initial state.
The CTQW is said to have PST from vertex $u\in V$ to vertex $v\in V$ at time $t_0$ if $\Phi_0=\delta_u$ and
$|\langle \delta_v, \Phi_{t_0}\rangle|=1$, where $\delta_u$ and $\delta_v$ denote the basis vectors in $l^2(V)$
associated with $u$ and $v$, respectively, and $\langle\cdot,\cdot\rangle$ means the inner product in $l^2(V)$.
\textbf{Just as pointed out in \cite{kempton}, perfect state transfer is a rather rare phenomenon,
and constructing examples of CTQWs when this occurs can be quite difficult.}

Bernoulli functionals are measurable functions on the Bernoulli space, which can also be viewed as functionals of a Bernoulli process.
Due to the variety in their algebraic and analytical structures, Bernoulli functionals have found wide application
in many problems in mathematical physics. For instance, the space $\mathfrak{h}$ of square integrable complex-valued Bernoulli functionals,
together with the annihilation and creation operators on it, can play an important role in describing the environment
of an open quantum system (see, e.g., \cite{chen-1,chen-2,wangcs-1,wangcs-2}).

In 2016, Wang and Ye \cite{wangcs-3} showed how to construct a model of DTQW (discrete-time quantum walk) on the integer lattice $\mathbb{Z}$
with $\mathfrak{h}$ as the coin space. In 2018, still on the integer lattice $\mathbb{Z}$, Wang et al. \cite{wangcs-4} introduced a model of open DTQW with
$\mathfrak{h}$ as the coin space. As is well known, however, DTQWs are quite different from CTQWs
although there is some correspondence between the two types of quantum walks \cite{childs,venegas}.

In this paper, motivated by the above, we would like to present a Bernoulli functional approach to CTQWs and PST.
More specifically, we would like to construct an abstract model of CTQW in the framework of $\mathfrak{h}$ and
investigate its PST property as well as its evolution behavior. Our main work is as follows.

First, we construct on a given subspace $\mathfrak{h}_L\subset \mathfrak{h}$ a self-adjoint operator $\Delta_L$, which we call the Laplacain on $\mathfrak{h}_L$,
via the canonical unitary involutions on $\mathfrak{h}$, and by analyzing its spectral structure we find out all eigenvalues of $\Delta_L$.
Then, we introduce an abstract model of CTQW with $\mathfrak{h}_L$ as its state space, which is governed by the Schr\"{o}dinger equation with $\Delta_L$ as the Hamiltonian,
and define its time-average probability distribution $\overline{P}$.
We obtain an explicit expression of $\overline{P}$, and, especially, we find $\overline{P}$ admits a symmetry property.
We justify the model by offering a graph-theoretic interpretation to the operator $\Delta_L$ as well as to the model itself.
Finally, we prove that the model has PST at time $t=\frac{\pi}{2}$. Some other interesting results are also proven.

The paper is organized as follows. In Section~\ref{sec-2}, we briefly recall the Bernoulli space, Bernoulli functionals
as well as some operators acting on these functionals. Section~\ref{sec-3} is our main work.
Here, in Subsection~\ref{subsec-3-1} we introduce our Laplacian $\Delta_L$ and analyze its spectral structure;
Subsection~\ref{subsec-3-2} describes our model of CTQW and examines properties of its time-average probability distribution;
In Subsection~\ref{subsec-3-3}, we justify our model by offering a graph-theoretic interpretation to the operator $\Delta_L$.
Finally in Subsection~\ref{subsec-3-4}, we discuss PST in our model.

\textbf{Frequently used notation.}\ \
Throughout this paper, $\mathbb{N}$ denotes the set of all nonnegative integers and $\Gamma$ its finite power set, namely
\begin{equation}\label{eq-1-1}
  \Gamma = \big\{\sigma \mid \sigma\subset \mathbb{N},\, \#(\sigma)<\infty\big\},
\end{equation}
where $\#(\sigma)$ means the cardinality of $\sigma$ as a set. For $\sigma\in \Gamma$ and $k\in \mathbb{N}$, $\sigma \setminus k \equiv \sigma\setminus \{k\}$,
and similarly $\sigma \cup k \equiv \sigma\cup \{k\}$.
If $L$ is a nonnegative integer, then $\Gamma_L$ stands for the power set of $\mathbb{N}_L$,
namely
\begin{equation}\label{eq-1-2}
\Gamma_L=\{\sigma \mid \sigma \subset \mathbb{N}_L\},
\end{equation}
where $\mathbb{N}_L\equiv\{0,1,\cdots, L\}$, the truncation of $\mathbb{N}$ by $L$.
Unless otherwise specified, letters like $j$, $k$ and $n$ are used to indicate integers.

\section{Bernoulli Functionals}\label{sec-2}

In this section, we briefly recall some necessary notions, notation as well as facts about Bernoulli functionals and operators acting on them.

Consider the set $\varOmega$ of all mappings $\omega\colon \mathbb{N}\rightarrow \{-1,1\}$, which is uncountable.
For each $n\geq 0$, let $\zeta_n\colon \varOmega\rightarrow \{-1,1\}$ be the canonical projection given by
\begin{equation}\label{eq-2-1}
  \zeta_n(\omega) = \omega(n),\quad \omega \in \varOmega.
\end{equation}
Then, one has a $\sigma$-algebra $\mathscr{F}=\sigma(\zeta_n;n\geq 0)$ over $\varOmega$, which is the one generated by the family of projections $(\zeta_n)_{n\geq 0}$.
Let $(\theta_n)_{n\geq 0}$ be a given sequence of real numbers with the property that $0 < \theta_n < 1$ for all $n\geq 0$.
By the well-known Kolmogorov's theorem \cite{cohn}, there exists a unique probability measure $P$  on $\mathscr{F}$ such that
\begin{equation}\label{eq-2-2}
P\circ(\zeta_{n_1}, \zeta_{n_2}, \cdots, \zeta_{n_k})^{-1}\big\{(\epsilon_1, \epsilon_2, \cdots, \epsilon_k)\big\}
=\prod_{j=1}^k \theta_{n_j}^{\frac{1+\epsilon_j}{2}}(1-\theta_{n_j})^{\frac{1-\epsilon_j}{2}},
\end{equation}
where $k$ is any positive integer, $\{n_j \mid 1\leq j \leq k\}\subset \mathbb{N}$ with $n_i\neq n_j$ (if $i\neq j$) and $\epsilon_j\in \{-1,1\}$, $1\le j \le k$.
Thus, one comes to a probability measure space\,
$(\varOmega, \mathscr{F}, P)$,
which is called the Bernoulli space.
Measurable functions on $(\varOmega, \mathscr{F}, P)$ are usually known as Bernoulli functionals.
In particular, square integrable functions on $(\varOmega, \mathscr{F}, P)$ are referred to as square integrable Bernoulli functionals

Denote by $\mathfrak{h}\equiv L^2(\varOmega, \mathscr{F}, P)$ the space of complex-valued square integrable Bernoulli functionals
with $\langle\cdot,\cdot\rangle$ being the usual inner product given by
\begin{equation}\label{eq-2-3}
  \langle\xi,\eta\rangle = \int_{\varOmega}\overline{\xi(\omega)}\eta(\omega)dP(\omega),\quad \xi,\, \eta\in\mathfrak{h},
\end{equation}
where $\overline{a}$ means the complex conjugate of a complex number $a$. The norm induced by $\langle\cdot,\cdot\rangle$ is written as $\|\cdot\|$.

Let $Z=(Z_n)_{n\geq 0}$ be the standardized sequence of the canonical projection sequence $(\zeta_n)_{n\geq 0}$, namely
\begin{equation}\label{eq-2-4}
   Z_n = \frac{\zeta_n + 1-2\theta_n }{2\sqrt{\theta_n(1-\theta_n)}},\quad n\geq 0.
\end{equation}
Then, $\mathfrak{h}$ has an orthonormal basis (ONB) of the form $\{Z_{\sigma} \mid \sigma \in \Gamma\}$, where
$Z_{\emptyset}=1$ and
\begin{equation}\label{eq-2-5}
Z_{\sigma} = \prod_{j\in \sigma}Z_j,\quad \sigma \in \Gamma,\, \sigma \neq \emptyset,
\end{equation}
where, as indicated above, $\Gamma$ is the finite power set of $\mathbb{N}$. Usually, $\{Z_{\sigma} \mid \sigma \in \Gamma\}$ is known
as the canonical ONB for $\mathfrak{h}$.

\begin{lemma}\label{lem-2-1}
Let $k\geq 0$. Then there exits a bounded operator $\Xi_k$ on $\mathfrak{h}$ such that
\begin{equation}\label{eq-2-6}
  \Xi_kZ_{\sigma} = \mathbf{1}_{\sigma}(k)Z_{\sigma\setminus k} + (1-\mathbf{1}_{\sigma}(k))Z_{\sigma\cup k},\quad \sigma \in \Gamma.
\end{equation}
\end{lemma}

\begin{proof}
Let $\partial_k$ be the annihilation operator on $\mathfrak{h}$ and $\partial_k^*$ its adjoint (see Section 2 of \cite{wangcs-3} for details).
Then, by putting $\Xi_k = \partial_k^* + \partial_k$, we find that $\Xi_k$ satisfies (\ref{eq-2-6}).
\end{proof}

\begin{lemma}\label{lem-2-2}
For all $k\geq 0$, $\Xi_k$ is a unitary involution, namely\, $\Xi_k^*=\Xi_k$  and\, $\Xi_k^2=I$.
\end{lemma}

\begin{proof}
From the proof of Lemma~\ref{lem-2-1}, we know that $\Xi_k = \partial_k^* + \partial_k$, which implies that $\Xi_k^*=\Xi_k$.
A straightforward calculation gives
\begin{equation*}
  \Xi_k^2 = (\partial_k^* + \partial_k)^2={\partial_k^*}^2 + \partial_k^*\partial_k + \partial_k\partial_k^* + \partial_k^2,
\end{equation*}
which, together with ${\partial_k^*}^2=\partial_k^2=0$ as well as the CAR in equal time (see Section 2 of \cite{wangcs-3} for details), implies $\Xi_k^2 =I$.
\end{proof}

The operator $\Xi_k$ is known as the canonical unitary involution on $\mathfrak{h}$ associated with $k$.
One typical property of the operator family $\{\Xi_k \mid k\geq 0\}$ is its commutativity, namely the following commutative relations hold true
\begin{equation}\label{eq-2-7}
\Xi_j\Xi_k = \Xi_k\Xi_j,\quad \forall\, j,\, k\geq 0.
\end{equation}
In view of the commutativity, one immediately finds the following products well-defined for all $\sigma\in \Gamma$
\begin{equation}\label{eq-2-8}
  \Xi_{\sigma}=\prod_{k\in \sigma}\Xi_k,
\end{equation}
where, by convention, $\Xi_{\emptyset} =I$. Using Lemma~\ref{lem-2-2} and the commutativity of $\{\Xi_k \mid k\geq 0\}$,
one can verify that $\Xi_{\sigma}$ remains a unitary involution for each $\sigma\in \Gamma$.

For a nonnegative integer $L\geq 0$, let $\mathfrak{h}_L= \mathrm{Span}\{Z_{\sigma} \mid \sigma \in \Gamma_L\}$,
the subspace of $\mathfrak{h}$ spanned by $\{Z_{\sigma} \mid \sigma \in \Gamma_L\}$.
Then, by inheriting the inner product $\langle\cdot,\cdot\rangle$ in $\mathfrak{h}$,
$\mathfrak{h}_L$ itself is a complex Hilbert space of dimension $2^{L+1}$.
The next lemma shows that $\Xi_k$ can be viewed as an operator on $\mathfrak{h}_L$ whenever $k\in \mathbb{N}_L$.

\begin{lemma}\label{lem-2-3}
Let $L\geq 0$ be a nonnegative integer. Then, for each $\sigma\in \Gamma_L$, $\Xi_{\sigma}$ leaves $\mathfrak{h}_L$ invariant.
In particular, for each $k\in \mathbb{N}_L$, $\Xi_k$ leaves $\mathfrak{h}_L$ invariant.
\end{lemma}

\begin{proof}
It suffices to show that $\Xi_kZ_{\tau} \in \mathfrak{h}_L$ for all $k\in \mathbb{N}_L$ and all $\tau\in \Gamma_L$.
In fact, for all $k\in \mathbb{N}_L$ and all $\tau\in \Gamma_L$,
both $\tau\setminus k$ and $\tau\cup k$ remain in $\Gamma_L$, which together with
\begin{equation*}
  \Xi_kZ_{\tau} = \mathbf{1}_{\tau}(k)Z_{\tau\setminus k} + (1-\mathbf{1}_{\tau}(k))Z_{\tau\cup k}
\end{equation*}
implies that $\Xi_kZ_{\tau}\in \mathfrak{h}_L$. This completes the proof.
\end{proof}

\begin{remark}\label{rem-2-1}
Let $L\geq 0$ be a nonnegative integer and $k\in \mathbb{N}_L$. Then, viewed as an operator on $\mathfrak{h}_L$,
$\Xi_k$ remains a unitary involution on $\mathfrak{h}_L$. This is an immediate consequence of Lemma~\ref{lem-2-2} and Lemma~\ref{lem-2-3}.
\end{remark}

\section{Main work}\label{sec-3}

This section presents our main work. In what follows, we always assume that $L\geq 0$ is fixed nonnegative integer.  We continue to write
\begin{equation}\label{eq-3-1}
\mathfrak{h}_L = \mathrm{Span}\{Z_{\sigma} \mid \sigma\in \Gamma_L\}.
\end{equation}
Note that $\mathfrak{h}_L$, together with the inner product $\langle\cdot,\cdot\rangle$ in $\mathfrak{h}$, is a complex Hilbert space of
dimension $2^{L+1}$. Clearly, the collection $\{Z_{\sigma} \mid \sigma\in \Gamma_L\}$ forms an ONB for $\mathfrak{h}_L$.
In the following, we refer to $\{Z_{\sigma} \mid \sigma\in \Gamma_L\}$ as the canonical ONB for $\mathfrak{h}_L$.

\subsection{Laplacian on $\mathfrak{h}_L$}\label{subsec-3-1}

In this subsection, we mainly introduce our Laplacian on $\mathfrak{h}_L$ and examine its spectral structure.

For each $k\in \mathbb{N}_L$, $\Xi_k$ can be viewed as an operator on $\mathfrak{h}_L$ since it leaves $\mathfrak{h}_L$ invariant.
We also use $I$ as the identity operator on $\mathfrak{h}_L$.
Consider the following operator on $\mathfrak{h}_L$
\begin{equation}\label{eq-3-2}
  \Delta_L = \sum_{k=0}^L (I-\Xi_k).
\end{equation}
According to Remark~\ref{rem-2-1}, $\Delta_L$ is a self-adjoint (bounded) operator on $\mathfrak{h}_L$.
Additionally, due to their unitary property and commutativity, $\{\Xi_k \mid k\in \mathbb{N}_L\}$ can naturally serve as shift operators on $\mathfrak{h}_L$.
In view of this, we call $\Delta_L$ the\textbf{ Laplacian} on $\mathfrak{h}_L$.

Now let us get an insight into the spectral structure of the operator $\Delta_L$. To this end, we first make some preparations.

\begin{proposition}\label{prop-3-1}
Define\, $\widehat{\Xi}^{(L)}_{\sigma}=\prod_{k=0}^L\big(I+\mathcal{E}_{\sigma}(k)\Xi_k\big)$ for $\sigma\in \Gamma_L$,
where $\mathcal{E}_{\sigma}(k)=2\mathbf{1}_{\sigma}(k)-1$. Then it holds true that
\begin{equation}\label{eq-3-3}
  \widehat{\Xi}^{(L)}_{\sigma}\widehat{\Xi}^{(L)}_{\tau}=
    \left\{
     \begin{array}{ll}
       2^{L+1}\widehat{\Xi}^{(L)}_{\sigma}, & \hbox{$\sigma$, $\tau\in \Gamma_L$,\, $\sigma=\tau$;} \\
       0, & \hbox{$\sigma$, $\tau\in \Gamma_L$,\, $\sigma\neq\tau$,}
     \end{array}
   \right.
\end{equation}
where $0$ means the null operator.
\end{proposition}

\begin{proof}
By using the commutativity of the operator system $\{\Xi_k \mid k\in \mathbb{N}_L\}$ on $\mathfrak{h}_L$ as well as the fact that
each $\Xi_k$ is a unitary involution
on $\mathfrak{h}_L$, we have
\begin{equation*}
  \big(\widehat{\Xi}^{(L)}_{\sigma}\big)^2
= \prod_{k=0}^L\big(I+\mathcal{E}_{\sigma}(k)\Xi_k\big)^2
= \prod_{k=0}^L2\big(I+\mathcal{E}_{\sigma}(k)\Xi_k\big)
= 2^{L+1}\widehat{\Xi}^{(L)}_{\sigma}.
\end{equation*}
Next, we consider the case of $\sigma\neq\tau$, which implies that $(\sigma\setminus \tau) \cup (\tau\setminus \sigma)\neq \emptyset$.
Take $j\in (\sigma\setminus \tau) \cup (\tau\setminus \sigma)$. Then
$\big(I+\mathcal{E}_{\sigma}(j)\Xi_k\big)\big(I+\mathcal{E}_{\tau}(j)\Xi_k\big)=0$, which implies that
\begin{equation*}
\begin{split}
 \widehat{\Xi}^{(L)}_{\sigma}\widehat{\Xi}^{(L)}_{\tau}
& = \prod_{k=0}^L\big(I+\mathcal{E}_{\sigma}(k)\Xi_k\big)\big(I+\mathcal{E}_{\tau}(k)\Xi_k\big)\\
& = \big(I+\mathcal{E}_{\sigma}(j)\Xi_k\big)\big(I+\mathcal{E}_{\tau}(j)\Xi_k\big)
    \prod_{k=0, k\neq j}^L\big(I+\mathcal{E}_{\sigma}(k)\Xi_k\big)\big(I+\mathcal{E}_{\tau}(k)\Xi_k\big)\\
& = 0.
\end{split}
\end{equation*}
Therefore (\ref{eq-3-3}) holds.
\end{proof}

\begin{theorem}\label{thr-3-2}
As a complex Hilbert space, $\mathfrak{h}_L$ has an ONB of the form $\big\{\widehat{Z}^{(L)}_{\sigma} \bigm| \sigma\in \Gamma_L\big\}$, where
\begin{equation}\label{eq-3-4}
\widehat{Z}^{(L)}_{\sigma}
= \frac{1}{\sqrt{2^{L+1}}}\sum_{\gamma \in \Gamma_L} (-1)^{\#(\gamma\setminus \sigma)}Z_{\gamma},
\end{equation}
where $\sum_{\gamma \in \Gamma_L}$ means to sum for all $\gamma\in \Gamma_L$.
\end{theorem}

\begin{proof}
Since the cardinality of the system $\big\{\widehat{Z}^{(L)}_{\sigma} \bigm| \sigma\in \Gamma_L\big\}$ coincides with
the dimension of $\mathfrak{h}_L$, it suffices to show that the system $\big\{\widehat{Z}^{(L)}_{\sigma} \bigm| \sigma\in \Gamma_L\big\}$
is orthonormal. Obviously, $\langle\widehat{Z}^{(L)}_{\sigma}, \widehat{Z}^{(L)}_{\sigma}\rangle =1$ for all $\sigma\in \Gamma_L$.
Now let $\sigma$, $\tau\in \Gamma_L$ with $\sigma\neq \tau$. Then, from the definition of $\widehat{Z}^{(L)}_{\sigma}$,
we find
\begin{equation*}
 \widehat{\Xi}^{(L)}_{\sigma}=\prod_{k=0}^L\big(I+\mathcal{E}_{\sigma}(k)\Xi_k\big)
= \sum_{\gamma\in \Gamma_L}\Big(\prod_{k\in \gamma} \mathcal{E}_{\sigma}(k)\Big)\Xi_{\gamma}
= \sum_{\gamma\in \Gamma_L}(-1)^{\#(\gamma\setminus \sigma)}\Xi_{\gamma},
\end{equation*}
which together with $\Xi_{\gamma}Z_{\emptyset} = Z_{\gamma}$ yields
\begin{equation*}
  \widehat{\Xi}^{(L)}_{\sigma}Z_{\emptyset}
= \sum_{\gamma\in \Gamma_L}(-1)^{\#(\gamma\setminus \sigma)}Z_{\gamma}
= \sqrt{2^{L+1}}\, \widehat{Z}^{(L)}_{\sigma}.
\end{equation*}
Similarly, we have $\widehat{\Xi}^{(L)}_{\tau}Z_{\emptyset} = \sqrt{2^{L+1}}\, \widehat{Z}^{(L)}_{\tau}$.
Hence, by using the self-adjoint property of $\widehat{\Xi}^{(L)}_{\sigma}$ as well as Proposition~\ref{prop-3-1}, we get
\begin{equation*}
  \big\langle \widehat{Z}^{(L)}_{\sigma}, \widehat{Z}^{(L)}_{\tau}\big\rangle
= \frac{1}{2^{L+1}}\big\langle \widehat{\Xi}^{(L)}_{\sigma}Z_{\emptyset}, \widehat{\Xi}^{(L)}_{\tau}Z_{\emptyset}\big\rangle
= \frac{1}{2^{L+1}}\big\langle Z_{\emptyset}, \widehat{\Xi}^{(L)}_{\sigma}\widehat{\Xi}^{(L)}_{\tau}Z_{\emptyset}\big\rangle
= 0.
\end{equation*}
This completes the proof.
\end{proof}

\begin{proposition}\label{prop-3-3}
Let $k\in \mathbb{N}_L$ and $ \sigma\in \Gamma_L$ be given. Then $\Xi_k\widehat{Z}^{(L)}_{\sigma} = \mathcal{E}_{\sigma}(k)\widehat{Z}^{(L)}_{\sigma}$.
\end{proposition}

\begin{proof}
By Lemma~\ref{lem-2-2}, we have
$\Xi_k(I+\mathcal{E}_{\sigma}(k)\Xi_k) = \mathcal{E}_{\sigma}(k) (I+\mathcal{E}_{\sigma}(k)\Xi_k)$, which together with Proposition~\ref{prop-3-1}
implies that
\begin{equation*}
\begin{split}
 \Xi_k\widehat{\Xi}^{(L)}_{\sigma}
& =\Xi_k (I+\mathcal{E}_{\sigma}(k)\Xi_k)\prod_{j=0,j\neq k }^L\big(I+\mathcal{E}_{\sigma}(j)\Xi_j\big)\\
& = \mathcal{E}_{\sigma}(k) (I+\mathcal{E}_{\sigma}(k)\Xi_k)\prod_{j=0,j\neq k }^L\big(I+\mathcal{E}_{\sigma}(j)\Xi_j\big)\\
& = \mathcal{E}_{\sigma}(k)\widehat{\Xi}^{(L)}_{\sigma}.
\end{split}
\end{equation*}
On the other hand, from the proof of Theorem~\ref{thr-3-2}, we find
$\widehat{Z}^{(L)}_{\sigma}=\frac{1}{\sqrt{2^{L+1}}}\widehat{\Xi}^{(L)}_{\sigma}Z_{\emptyset}$. Thus, we finally have
\begin{equation*}
  \Xi_k\widehat{Z}^{(L)}_{\sigma}
  = \frac{1}{\sqrt{2^{L+1}}}\Xi_k\widehat{\Xi}^{(L)}_{\sigma}Z_{\emptyset}
  = \frac{1}{\sqrt{2^{L+1}}}\mathcal{E}_{\sigma}(k)\widehat{\Xi}^{(L)}_{\sigma}Z_{\emptyset}
  = \mathcal{E}_{\sigma}(k)\widehat{Z}^{(L)}_{\sigma}.
\end{equation*}
\end{proof}

\begin{theorem}\label{thr-3-4}
As an operator on $\mathfrak{h}_L$, $\Delta_L$ admits a diagonal representation of the following form
\begin{equation}\label{eq-3-5}
  \Delta_L = \sum_{\sigma\in \Gamma_L} 2\big(L+1 - \#(\sigma)\big)\big|\widehat{Z}^{(L)}_{\sigma}\big\rangle\!\big\langle \widehat{Z}^{(L)}_{\sigma}\big|,
\end{equation}
where $\big|\widehat{Z}^{(L)}_{\sigma}\big\rangle\!\big\langle \widehat{Z}^{(L)}_{\sigma}\big|$ denotes the Dirac operator associated
with vector $\widehat{Z}^{(L)}_{\sigma}$.
\end{theorem}

\begin{proof}
Let $\sigma\in \Gamma_L$ be given. Then, by using the definition of $\Delta_L$ given in (\ref{eq-3-2}), and Proposition~\ref{prop-3-3}, we have
\begin{equation*}
  \Delta_L \widehat{Z}^{(L)}_{\sigma}
= \sum_{k=0}^L (I-\Xi_k)\widehat{Z}^{(L)}_{\sigma}
= \sum_{k=0}^L \big(1-\mathcal{E}_{\sigma}(k)\big)\widehat{Z}^{(L)}_{\sigma}
=\Big( L+1 - \sum_{k=0}^L\mathcal{E}_{\sigma}(k)\Big)\widehat{Z}^{(L)}_{\sigma},
\end{equation*}
which together with $\sum_{k=0}^L\mathcal{E}_{\sigma}(k) = 2\#(\sigma)-(L+1)$ gives
\begin{equation}\label{eq-3-6}
  \Delta_L \widehat{Z}^{(L)}_{\sigma} = 2\big(L+1-\#(\sigma)\big)\widehat{Z}^{(L)}_{\sigma}.
\end{equation}
Thus, for any $\xi \in \mathfrak{h}_L$, by using Theorem~\ref{thr-3-2} we come to
\begin{equation*}
\begin{split}
  \Delta_L\xi
  & = \sum_{\sigma\in \Gamma_L}\big\langle \widehat{Z}^{(L)}_{\sigma}, \xi\big\rangle \Delta_L\widehat{Z}^{(L)}_{\sigma}\\
  & = \sum_{\sigma\in \Gamma_L}2\big(L+1-\#(\sigma)\big)\big\langle \widehat{Z}^{(L)}_{\sigma}, \xi\big\rangle \widehat{Z}^{(L)}_{\sigma}\\
  & = \sum_{\sigma\in \Gamma_L}2\big(L+1-\#(\sigma)\big)\big|\widehat{Z}^{(L)}_{\sigma}\big\rangle\!\big\langle \widehat{Z}^{(L)}_{\sigma}\big|\xi,
\end{split}
\end{equation*}
which together with the arbitrariness of $\xi \in \mathfrak{h}_L$ gives (\ref{eq-3-5}).
\end{proof}

We are now ready precisely to describe the spectral structure of the Lalplacian $\Delta_L$ as follows.

\begin{corollary}
As an operator on $\mathfrak{h}_L$,  $\Delta_L$ has a spectrum $\mathrm{Spec}(\Delta_L)$ of the following form
\begin{equation}\label{eq-3-7}
\mathrm{Spec}(\Delta_L)=\big\{2k \mid k\in \mathbb{N}_L\big\}.
\end{equation}
Moreover, for each $k\in \mathbb{N}_L$, the eigen-space corresponding $2k$ is exactly the one spanned by vectors
$\big\{\widehat{Z}^{(L)}_{\sigma} \mid \sigma\in \Gamma_L,\, \#(\sigma) = L+1-k\big\}$.
\end{corollary}

\subsection{Abstract model of CTQW with $\mathfrak{h}_L$ being the state space}\label{subsec-3-2}

In this subsection, we establish our abstract model of CTQW and examine its evolution properties.
As usual, we use $\mathbb{R}$ to mean the set of all real numbers.

Recall that $\Gamma_L$ denotes the power set of $\mathbb{N}_L$, namely $\Gamma_L=\{\sigma \mid \sigma \subset \mathbb{N}_L\}$.
For the sake of clarity, we call elements of $\Gamma_L$ nodes. For two nodes $\sigma$, $\tau\in \Gamma_L$, we define
\begin{equation}\label{eq-3-8}
\sigma\bigtriangleup\tau = (\sigma\setminus\tau)\cup(\tau\setminus\sigma),
\end{equation}
which is usually known as the symmetric difference between $\sigma$ and $\tau$. Note that $\Gamma_L$ together with $\bigtriangleup$ forms a algebraic group.
Recall also that $\{Z_{\sigma} \mid \sigma \in \Gamma_L\}$ is called the canonical ONB for $\mathfrak{h}_L$.

\begin{proposition}\label{prop-3-6}
Let $\sigma$, $\tau\in \Gamma_L$ be given. Then it holds true that
\begin{equation}\label{eq-3-9}
  \big\langle Z_{\sigma}, \widehat{Z}^{(L)}_{\tau}\big\rangle
= \big\langle \widehat{Z}^{(L)}_{\tau}, Z_{\sigma}\big\rangle
=\frac{1}{\sqrt{2^{L+1}}}(-1)^{\#(\sigma\setminus \tau)}.
\end{equation}
\end{proposition}

\begin{proof}
According to Theorem~\ref{thr-3-2}, $\widehat{Z}^{(L)}_{\tau}= \frac{1}{\sqrt{2^{L+1}}}\sum_{\gamma\in \Gamma_L}(-1)^{\#(\gamma\setminus \tau)}Z_{\gamma}$.
Thus, by using the orthogonality of the basis vectors, we come to
\begin{equation*}
  \big\langle Z_{\sigma}, \widehat{Z}^{(L)}_{\tau}\big\rangle = \frac{1}{\sqrt{2^{L+1}}}(-1)^{\#(\sigma\setminus \tau)}\langle Z_{\sigma},Z_{\sigma}\rangle
  = \frac{1}{\sqrt{2^{L+1}}}(-1)^{\#(\sigma\setminus \tau)}.
\end{equation*}
Similarly, we can show that $\big\langle \widehat{Z}^{(L)}_{\tau}, Z_{\sigma}\big\rangle$ has the same expression as
$\big\langle Z_{\sigma}, \widehat{Z}^{(L)}_{\tau}\big\rangle$.
\end{proof}

\begin{definition}\label{def-3-1}
The continuous-time quantum walk governed by $\Delta_L$ (briefly, the walk $\Delta_L$ or the walk) admits the following features.
\begin{itemize}
  \item[(1)]\  The state space of the walk is $\mathfrak{h}_L$ and its states are represented by unit vectors in $\mathfrak{h}_L$;
  \item[(2)]\  The evolution of the walk is governed by equation
    \begin{equation}\label{eq-3-10}
        \xi_t = \mathrm{e}^{\mathrm{i}t\Delta_L}\xi_0,\quad t\in \mathbb{R},
    \end{equation}
      where $\xi_t$ denotes its state at time $t$, especially $\xi_0$ is its initial state;
  \item[(3)]\ The probability $P_t(\sigma)$ that the walker is found at node $\sigma\in \Gamma_L$ at time $t\in \mathbb{R}$ is given by
   \begin{equation}\label{eq-3-11}
   P_t(\sigma) = |\langle Z_{\sigma}, \xi_t\rangle|^2.
   \end{equation}
\end{itemize}
In that case, the collection $\{\,\xi_t \mid t\in \mathbb{R}\,\}$ is called the trajectory of the walk with initial state $\xi_0$,
while, for $t\in \mathbb{R}$, the function $\sigma\mapsto P_t(\sigma)$ on $\Gamma_L$ is referred to as the probability distribution of the walk at time $t$.
\end{definition}

Clearly, each trajectory of the walk $\Delta_L$ is completely determined by the operator $\Delta_L$ together with the initial state $\xi_0$.
The next theorem shows that the evolution of the walk $\Delta_L$ is even periodic in time.

\begin{theorem}\label{thr-periodicity}
It holds true that\, $\mathrm{e}^{\mathrm{i}(t+\pi)\Delta_L} = \mathrm{e}^{\mathrm{i}t\Delta_L}$, $\forall\,t\in \mathbb{R}$.
In particular, each trajectory $\{\,\xi_t \mid t\in \mathbb{R}\,\}$ of the walk $\Delta_L$ admits periodicity of the form
\begin{equation}\label{eq-periodicity}
  \xi_{t +\pi} = \xi_t,\quad \forall\,t\in \mathbb{R}.
\end{equation}
\end{theorem}

\begin{proof}
(\ref{eq-periodicity}) is a direct consequence of the equality $\mathrm{e}^{\mathrm{i}(t+\pi)\Delta_L} = \mathrm{e}^{\mathrm{i}t\Delta_L}$
together with (\ref{eq-3-10}). Next, we show that $\mathrm{e}^{\mathrm{i}(t+\pi)\Delta_L} = \mathrm{e}^{\mathrm{i}t\Delta_L}$.
However, since $\mathrm{e}^{\mathrm{i}(t+\pi)\Delta_L}= \mathrm{e}^{\mathrm{i}t\Delta_L}\mathrm{e}^{\mathrm{i}\pi\Delta_L}$,
it suffices to verify $\mathrm{e}^{\mathrm{i}\pi\Delta_L}=I$.
In fact, by using Theorem~\ref{thr-3-4} and Theorem~\ref{thr-3-2}, we immediately get
\begin{equation*}
 \mathrm{e}^{\mathrm{i}\pi\Delta_L}
  = \sum_{\sigma\in \Gamma_L} \mathrm{e}^{2\mathrm{i}(L+1 - \#(\sigma))\pi}\big|\widehat{Z}^{(L)}_{\sigma}\big\rangle\!\big\langle \widehat{Z}^{(L)}_{\sigma}\big|
  = \sum_{\sigma\in \Gamma_L} \big|\widehat{Z}^{(L)}_{\sigma}\big\rangle\!\big\langle \widehat{Z}^{(L)}_{\sigma}\big|
  = I.
\end{equation*}
Here the identity $\mathrm{e}^{2\mathrm{i}(L+1 - \#(\sigma))\pi}=1$ is used.
\end{proof}

For the walk $\Delta_L$, its probability distributions depend on its initial state $\xi_0$ in general.
From a physical point of view, however, the walk should start from node $\emptyset\in \Gamma_L$,
namely its initial state $\xi_0$ should be such that $\xi_0=Z_{\emptyset}$,
where $Z_{\emptyset}$ is the basis vector of the canonical ONB associated with node $\emptyset$,
which represents a vacuum state in physics.

\begin{theorem}\label{thr-3-8}
Consider the walk $\Delta_L$.
Let its initial state $\xi_0$ be such that $\xi_0=Z_{\emptyset}$.
Then, at time $t\in \mathbb{R}$, its probability distribution admits a representation of the following form
\begin{equation}\label{eq-3-13}
  P_t(\sigma) = \frac{1}{4^{L+1}}\Big|\sum_{\gamma \in \Gamma_L} (-1)^{\#(\sigma\setminus \gamma)}\mathrm{e}^{2\mathrm{i}(L+1-\#(\gamma))t}\Big|^2,\quad
  \sigma\in \Gamma_L.
\end{equation}
In particular, at time $t=\frac{\pi}{2}$, its probability distribution even degenerates into a single-node distribution of the form
\begin{equation}\label{eq}
  P_{\frac{\pi}{2}}(\sigma) =
\left\{
  \begin{array}{ll}
    1, & \hbox{$\sigma=\mathbb{N}_L$;} \\
    0, & \hbox{$\sigma\in \Gamma_L$, $\sigma\neq \mathbb{N}_L$.}
  \end{array}
\right.
\end{equation}
\end{theorem}

\begin{proof}
Let $t\in \mathbb{R}$ be given. Then, it follows from Theorem~\ref{thr-3-4} that $\mathrm{e}^{\mathrm{i}t\Delta_L}$ has a diagonal representation of the form
\begin{equation}\label{unitary-diagonal}
 \mathrm{e}^{\mathrm{i}t\Delta_L}
= \sum_{\gamma\in \Gamma_L} \mathrm{e}^{2\mathrm{i}(L+1 - \#(\gamma))t}\big|\widehat{Z}^{(L)}_{\gamma}\big\rangle\!\big\langle \widehat{Z}^{(L)}_{\gamma}\big|,
\end{equation}
which, together with $\xi_0 = Z_{\emptyset}$ as well as Proposition~\ref{prop-3-6}, gives
\begin{equation*}
\xi_t
 = \mathrm{e}^{\mathrm{i}t\Delta_L}\xi_0
 = \sum_{\gamma\in \Gamma_L} \mathrm{e}^{2\mathrm{i}(L+1 - \#(\gamma))t}\big\langle \widehat{Z}^{(L)}_{\gamma},Z_{\emptyset}\big\rangle\widehat{Z}^{(L)}_{\gamma}
 = \frac{1}{\sqrt{2^{L+1}}}\sum_{\gamma\in \Gamma_L} \mathrm{e}^{2\mathrm{i}(L+1 - \#(\gamma))t}\widehat{Z}^{(L)}_{\gamma}.
\end{equation*}
Thus, for $\sigma\in \Gamma_L$, by using Proposition~\ref{prop-3-6} we have
\begin{equation*}
\langle Z_{\sigma},\xi_t\rangle
= \frac{1}{\sqrt{2^{L+1}}}\sum_{\gamma\in \Gamma_L} \mathrm{e}^{2\mathrm{i}(L+1 - \#(\gamma))t} \big\langle Z_{\sigma},\widehat{Z}^{(L)}_{\gamma}\big\rangle
= \frac{1}{2^{L+1}}\sum_{\gamma\in \Gamma_L} \mathrm{e}^{2\mathrm{i}(L+1 - \#(\gamma))t} (-1)^{\#(\sigma\setminus \gamma)},
\end{equation*}
which implies that
\begin{equation*}
   P_t(\sigma)
   = |\langle Z_{\sigma},\xi_t\rangle|^2
   = \frac{1}{4^{L+1}}\Big|\sum_{\gamma\in \Gamma_L} (-1)^{\#(\sigma\setminus \gamma)}\mathrm{e}^{2\mathrm{i}(L+1 - \#(\gamma))t} \Big|^2.
\end{equation*}
Finally, by a straightforward calculation, we find
\begin{equation*}
\begin{split}
  P_{\frac{\pi}{2}}(\mathbb{N}_L)
  & = \frac{1}{4^{L+1}}\Big|\sum_{\gamma\in \Gamma_L} (-1)^{L+1-\#(\gamma)}\mathrm{e}^{\mathrm{i}(L+1 - \#(\gamma))\pi} \Big|^2\\
  &= \frac{1}{4^{L+1}}\Big|\sum_{\gamma\in \Gamma_L} (-1)^{L+1-\#(\gamma)}(-1)^{L+1-\#(\gamma)} \Big|^2\\
  &= \frac{1}{4^{L+1}}\Big|\sum_{\gamma\in \Gamma_L} 1 \Big|^2\\
  &=1,
\end{split}
\end{equation*}
which, together with the fact that the function $\sigma\mapsto P_{\frac{\pi}{2}}(\sigma)$ is a probability distribution on $\Gamma_L$,
implies that $P_{\frac{\pi}{2}}(\sigma)=0$ for all $\sigma\in \Gamma_L$ with $\sigma\neq \mathbb{N}_L$.
\end{proof}

Recall that $P_t(\sigma)$ is the probability that walker is found at node $\sigma\in \Gamma_L$ at time $t\in \mathbb{R}$.
Now consider the function $t\mapsto P_t(\sigma)$ for given $\sigma\in \Gamma_L$.
According to Theorem~\ref{thr-periodicity}, this is a periodic function on $\mathbb{R}$ with $\pi$ being  a period.
In view of this fact, we naturally introduce a function $\overline{P}(\cdot)$ on $\Gamma_L$ in the following manner
\begin{equation}\label{eq-time-average distribution}
 \overline{P}(\sigma) = \frac{1}{\pi}\int_0^{\pi}P_t(\sigma)dt,\quad \sigma\in \Gamma_L,
\end{equation}
and call it the \textbf{time-average probability distribution} of the walk $\Delta_L$.
As can be seen, the function $t\mapsto P_t(\sigma)$ is continuous for each $\sigma\in \Gamma_L$, which implies that $\overline{P}(\sigma)$ is well-defined.
Thus, $\overline{P}(\cdot)$ makes sense as a probability distribution on $\Gamma_L$.

\begin{theorem}\label{thr-3-9}
Let the initial state $\xi_0$ of the walk $\Delta_L$ be such that $\xi_0=Z_{\emptyset}$.
Then
\begin{equation}\label{eq-3-17}
  \overline{P}(\sigma)
  = \frac{1}{4^{L+1}}\sum_{(\gamma_1,\gamma_2)}(-1)^{\#(\sigma\setminus \gamma_1) + \#(\sigma\setminus \gamma_2)},\quad \sigma\in \Gamma_L,
\end{equation}
where $\sum_{(\gamma_1,\gamma_2)}$ means to sum over the set $\big\{(\gamma_1,\gamma_2)\in \Gamma_L\times \Gamma_L \mid \#(\gamma_1)=\#(\gamma_2)\big\}$.
\end{theorem}

\begin{proof}
Let $\sigma\in \Gamma_L$ be given. Then, by writing
\begin{equation*}
\Sigma_1=\big\{(\gamma_1,\gamma_2) \in \Gamma_L\times \Gamma_L \mid \#(\gamma_1)=\#(\gamma_2)\big\}
\end{equation*}
and $\Sigma_2=(\Gamma_L\times \Gamma_L)\setminus \Sigma_1$, we can rewrite (\ref{eq-3-13}) as
\begin{equation*}
\begin{split}
  P_t(\sigma)
   &= \frac{1}{4^{L+1}}\sum_{\gamma_1 \in \Gamma_L} (-1)^{\#(\sigma\setminus \gamma_1)}\mathrm{e}^{-2\mathrm{i}(L+1-\#(\gamma_1))t}
      \sum_{\gamma_2 \in \Gamma_L} (-1)^{\#(\sigma\setminus \gamma_2)}\mathrm{e}^{2\mathrm{i}(L+1-\#(\gamma_2))t}\\
   &= \frac{1}{4^{L+1}}\Big[\sum_{(\gamma_1,\gamma_2)\in \Sigma_1}\!\!(-1)^{\#(\sigma\setminus \gamma_1)+ \#(\sigma\setminus \gamma_2)}
    + \sum_{(\gamma_1,\gamma_2)\in \Sigma_2}\!\! (-1)^{\#(\sigma\setminus \gamma_1)+ \#(\sigma\setminus \gamma_2)}
         \mathrm{e}^{2\mathrm{i}(\#(\gamma_1)- \#(\gamma_2))t}\Big],
\end{split}
\end{equation*}
where $t\in \mathbb{R}$. Clearly, $\int_0^{\pi}\mathrm{e}^{2\mathrm{i}(\#(\gamma_1)- \#(\gamma_2))t}dt= 0$ for all $(\gamma_1,\gamma_2)\in \Sigma_2$.
Thus
\begin{equation*}
  \overline{P}(\sigma)
  =\frac{1}{\pi} \int_0^{\pi}P_t(\sigma)dt
  = \frac{1}{4^{L+1}}\sum_{(\gamma_1,\gamma_2)\in \Sigma_1}\!\!(-1)^{\#(\sigma\setminus \gamma_1)+ \#(\sigma\setminus \gamma_2)},
\end{equation*}
which is exactly the same as (\ref{eq-3-17}).
\end{proof}

Two nodes $\sigma$, $\tau\in \Gamma_L$ are said to be symmetric with respect to $\mathbb{N}_L$
if they satisfy that
\begin{equation*}
\sigma \cap \tau = \emptyset\quad \mbox{and} \quad\sigma\cup \tau= \mathbb{N}_L,
\end{equation*}
which is equivalent to that $\tau = \mathbb{N}_L\setminus \sigma$.
A function $Q(\cdot)$ on $\Gamma_L$ is called symmetric if it satisfies the condition given below
\begin{equation}\label{eq}
  Q(\sigma) = Q(\sigma^c),\quad \sigma\in \Gamma_L,
\end{equation}
where $\sigma^c=\mathbb{N}_L\setminus \sigma$. The next theorem characterizes symmetry of the walk $\Delta_L$
through its time-average probability distribution.

\begin{theorem}\label{thr-3-10}
Let the initial state $\xi_0$ of the walk $\Delta_L$ be such that $\xi_0=Z_{\emptyset}$.
Then, its time-average probability distribution $\overline{P}(\cdot)$ is symmetric on $\Gamma_L$.
\end{theorem}

\begin{proof}
Let $\sigma\in \Gamma_L$ be given. Then, by Theorem~\ref{thr-3-9}, we have
\begin{equation*}
  \overline{P}(\sigma) = \sum_{(\gamma_1,\gamma_2)\in \Sigma_1}(-1)^{\#(\sigma\setminus \gamma_1) + \#(\sigma\setminus \gamma_2)},\quad
   \overline{P}(\sigma^c) = \sum_{(\gamma_1,\gamma_2)\in \Sigma_1}(-1)^{\#(\sigma^c\setminus \gamma_1) + \#(\sigma^c\setminus \gamma_2)},
\end{equation*}
where $\Sigma_1=\big\{(\gamma_1,\gamma_2)\in \Gamma_L\times \Gamma_L \mid \#(\gamma_1)=\#(\gamma_2)\big\}$.
On the other hand, for each $(\gamma_1,\gamma_2)\in \Sigma_1$, careful calculations give
\begin{equation*}
\begin{split}
(-1)^{\#(\sigma^c\setminus \gamma_1) + \#(\sigma^c\setminus \gamma_2)}
  &= (-1)^{\#(\mathbb{N}_L\setminus (\sigma\cup \gamma_1)) + \#(\mathbb{N}_L\setminus (\sigma\cup \gamma_2))}\\
  &= (-1)^{2(L+1)- \#(\sigma\cup \gamma_1)-\#(\sigma\cup \gamma_2)}\\
  &= (-1)^{\#(\sigma\cup \gamma_1)+\#(\sigma\cup \gamma_2)}\\
  &= (-1)^{ \#(\gamma_1)+ \#(\gamma_2) + \#(\sigma\setminus \gamma_1)+\#(\sigma\setminus \gamma_2)}\\
  &= (-1)^{\#(\sigma\setminus \gamma_1)+\#(\sigma\setminus \gamma_2)}.
\end{split}
\end{equation*}
Thus, combining all the above, we come to
\begin{equation*}
  \overline{P}(\sigma)
= \sum_{(\gamma_1,\gamma_2)\in \Sigma_1}(-1)^{\#(\sigma\setminus \gamma_1) + \#(\sigma\setminus \gamma_2)}
= \sum_{(\gamma_1,\gamma_2)\in \Sigma_1}(-1)^{\#(\sigma^c\setminus \gamma_1) + \#(\sigma^c\setminus \gamma_2)}
= \overline{P}(\sigma^c),
\end{equation*}
which, together with the arbitrariness of $\sigma\in \Gamma_L$, implies that $\overline{P}(\cdot)$ is symmetric.
\end{proof}

\begin{proposition}\label{prop-3-11}
Let the initial state $\xi_0$ of the walk $\Delta_L$ be such that $\xi_0=Z_{\emptyset}$.
Then
\begin{equation}\label{eq-3-19}
 \overline{P}(\emptyset)  = \overline{P}(\mathbb{N}_L) = \frac{(2L+1)!!}{(2L+2)!!},
\end{equation}
where $(2L+1)!!= (2L+1)\times(2L-1)\times\cdots \times1$ and $(2L+2)!!$ has the similar meaning.
\end{proposition}

\begin{proof}
Write $\Sigma_1= \big\{(\gamma_1,\gamma_2)\in \Gamma_L\times \Gamma_L \mid \#(\gamma_1)= \#(\gamma_2)\big\}$.
Then, a careful calculation gives
\begin{equation}\label{eq-3-20}
\#(\Sigma_1)= \sum_{k=0}^{L+1} \binom{L+1}{k}\binom{L+1}{k}
= \binom{2L+2}{L+1}.
\end{equation}
Now, by using Theorem~\ref{thr-3-9}, we have
\begin{equation*}
  \overline{P}(\emptyset)
   = \frac{1}{4^{L+1}} \sum_{(\gamma_1,\gamma_2)\in \Sigma_1}(-1)^{\#(\emptyset\setminus \gamma_1)+\#(\emptyset\setminus \gamma_2)}
   = \frac{1}{4^{L+1}} \sum_{(\gamma_1,\gamma_2)\in \Sigma_1}(-1)^0
   = \frac{\#(\Sigma_1)}{4^{L+1}},
\end{equation*}
which together with (\ref{eq-3-20}) yields
\begin{equation*}
  \overline{P}(\emptyset) = \frac{1}{4^{L+1}}\binom{2L+2}{L+1} = \frac{(2L+1)!!}{(2L+2)!!}.
\end{equation*}
Finally, by Theorem~\ref{thr-3-10}, we know that $\overline{P}(\mathbb{N}_L) =\overline{P}(\emptyset)$.
\end{proof}

\subsection{Graph-theoretic interpretation}\label{subsec-3-3}

In the present subsection, we justify our model by giving a graph-theoretic interpretation to the operator $\Delta_L$
as well as to the model itself.

Consider the set $\Gamma_L$, which is the power set of $\mathbb{N}_L\equiv \{0,1,\cdots,L\}$.
Recall that, for $\sigma$, $\tau\in \Gamma_L$, $\sigma\bigtriangleup\tau$ means the symmetric difference between them.

\begin{definition}
Let $\sigma$, $\tau\in \Gamma_L$. $\sigma$ is said to be adjacent to $\tau$ if $\#(\sigma\bigtriangleup\tau) =1$. In that case,
we write $\sigma\sim \tau$.
\end{definition}

Clearly, $(\Gamma_L,\sim)$ forms a graph, where the vertex set is $\Gamma_L$ itself, while the edge set $E_L$ is given by
\begin{equation}\label{eq}
E_L= \big\{\{\sigma,\tau\} \mid \sigma,\tau\in \Gamma_L,\, \sigma\sim \tau\big\}.
\end{equation}
In the following, we simply call $(\Gamma_L,\sim)$ the graph $\Gamma_L$.

\begin{proposition}\label{prop-3-12}
Let $\sigma$, $\tau$ be vertices of the graph $\Gamma_L$. Then $\sigma\sim \tau$ if and only if there exits a unique $k\in \mathbb{N}_L$ such that
\begin{equation}\label{eq-3-22}
  \tau =\left\{
          \begin{array}{ll}
            \sigma\setminus k, & \hbox{$k\in \sigma$;} \\
            \sigma\cup k, & \hbox{$k\notin \sigma$,}
          \end{array}
        \right.
\end{equation}
where $\sigma\setminus k= \sigma\setminus \{k\}$ and $\sigma\cup k=\sigma\cup \{k\}$.
\end{proposition}

\begin{proof}
The ``if'' part is easy to verify. Next, we verify the ``only if'' part. Suppose $\sigma\sim \tau$. Then, by the definition,
\begin{equation*}
\#(\sigma\setminus \tau)+\#(\tau\setminus\sigma)= \#(\sigma\bigtriangleup\tau) =1,
\end{equation*}
which implies that $\#(\sigma\setminus \tau)=1$ with $\#(\tau\setminus\sigma)=0$ or $\#(\sigma\setminus \tau)=0$ with $\#(\tau\setminus\sigma)=1$,
which implies that there exists $k\in \mathbb{N}_L$ such that
\begin{equation*}
\sigma\setminus \tau=\{k\}\ \ \mbox{with}\ \ \tau \subset \sigma\quad
\mbox{or}\quad
  \sigma\subset \tau\ \ \mbox{with}\ \ \tau\setminus \sigma =\{k\},
\end{equation*}
which is equivalent to (\ref{eq-3-22}). The uniqueness of the above $k$ follows from $\sigma\bigtriangleup\tau=\{k\}$.
\end{proof}

The above proposition characterizes the adjacency relation in the graph $\Gamma_L$.
The next proposition further shows that the graph $\Gamma_L$ is regular and its degree is $L+1$.

\begin{proposition}
Let $\sigma\in\Gamma_L$ and write $\mathcal{N}(\sigma)=\{\tau \in \Gamma_L \mid \tau\sim \sigma\}$. Then it holds that
\begin{equation}\label{eq-3-23}
  \mathcal{N}(\sigma)=\big\{\sigma\setminus k \mid k\in \sigma\big\}\cup \big\{\sigma\cup k \mid k\in \mathbb{N}_L\setminus \sigma\big\}.
\end{equation}
In particular, $\#(\mathcal{N}(\sigma))=L+1$, namely the degree of the vertex $\sigma$ is $L+1$.
\end{proposition}

\begin{proof}
Obviously, (\ref{eq-3-23}) holds for $\sigma=\emptyset$ or $\sigma=\mathbb{N}_L$. Next, we suppose $\sigma\neq \emptyset$ and $\sigma\neq \mathbb{N}_L$.
In this case, by Proposition~\ref{prop-3-12}, we immediately have
\begin{equation*}
  \mathcal{N}(\sigma)\supset \big\{\sigma\setminus k \mid k\in \sigma\big\}\cup \big\{\sigma\cup k \mid k\in \mathbb{N}_L\setminus \sigma\big\}.
\end{equation*}
On the other hand, if $\tau\in \mathcal{N}(\sigma)$, then again by Proposition~\ref{prop-3-12}, there exits $k\in \sigma$ such that $\tau=\sigma\setminus k$
or there exists $k\in \mathbb{N}_L\setminus \sigma$ such that $\tau=\sigma\cup k$, which implies that
\begin{equation*}
\tau\in \big\{\sigma\setminus k \mid k\in \sigma\big\}\cup \big\{\sigma\cup k \mid k\in \mathbb{N}_L\setminus \sigma\big\}.
\end{equation*}
Thus $\mathcal{N}(\sigma)\subset \big\{\sigma\setminus k \mid k\in \sigma\big\}\cup \big\{\sigma\cup k \mid k\in \mathbb{N}_L\setminus \sigma\big\}$.
In summary, (\ref{eq-3-23}) remains true for the case of $\sigma\neq \emptyset$ and $\sigma\neq \mathbb{N}_L$.
\end{proof}

Let $C(\Gamma_L)$ be the set of all functions $f\colon \Gamma_L\rightarrow \mathbb{C}$, where $\mathbb{C}$ denotes the set of all complex numbers.
It is known that $C(\Gamma_L)$, together with the usual addition and scalar multiplication as well as the usual inner product
$\langle\cdot,\cdot\rangle_{C(\Gamma_L)}$, forms a complex Hilbert space.
Moreover, $C(\Gamma_L)$ has an orthonormal basis of the form $\{e_{\sigma} \mid \sigma \in \Gamma_L\}$, where $e_{\sigma}$ is the function on $\Gamma_L$ given by
\begin{equation*}
  e_{\sigma}(\gamma)=
    \left\{
     \begin{array}{ll}
       1, & \hbox{$\gamma=\sigma$;} \\
       0, & \hbox{$\gamma\neq \sigma$, $\gamma\in \Gamma_L$.}
     \end{array}
   \right.
\end{equation*}
Thus, as a complex Hilbert space, $C(\Gamma_L)$ is of dimension $2^{L+1}$.
Recall that $\mathfrak{h}_L$ is also a complex Hilbert space of dimension $2^{L+1}$. Therefore, there exists a unitary isomorphism
$\mathsf{F}\colon C(\Gamma_L)\rightarrow \mathfrak{h}_L$ such that
\begin{equation*}
  \mathsf{F}e_{\sigma} = Z_{\sigma},\quad \sigma \in \Gamma_L,
\end{equation*}
where $Z_{\sigma}$ is the basis vector of the canonical ONB for $\mathfrak{h}_L$.

\begin{definition}\label{def-graph-Laplacian}
The Laplacian $\widetilde{\Delta}_L$ of the graph $\Gamma_L$ is the operator on $C(\Gamma_L)$ defined by
\begin{equation}\label{eq-graph-Laplacian}
  [\widetilde{\Delta}_Lf](\sigma)
= \sum_{\tau\in \mathcal{N}(\sigma)} [f(\sigma)-f(\tau)],\quad \sigma \in \Gamma_L,
\end{equation}
where $f\in C(\Gamma_L)$.
\end{definition}

According to the general theory of spectral graph (see, e.g., \cite{obata}), the Laplacian $\widetilde{\Delta}_L$
is a self-adjoint operator on $C(\Gamma_L)$.
Carefully checking (\ref{eq-graph-Laplacian}) and (\ref{eq-3-23}), we see that $\widetilde{\Delta}_L$ can be equivalently defined as
\begin{equation}\label{eq-graph-Laplacian-2}
  [\widetilde{\Delta}_Lf](\sigma)
= (L+1)f(\sigma)- \sum_{k=0}^L \big[\mathbf{1}_{\sigma}(k)f(\sigma\setminus k) + (1-\mathbf{1}_{\sigma}(k))f(\sigma\cup k)\big],\quad \sigma \in \Gamma_L,
\end{equation}
where $f\in C(\Gamma_L)$.
Next, we examine links between the Laplacian $\widetilde{\Delta}_L$ of
the graph $\Gamma_L$ and the Laplacian $\Delta_L$ on $\mathfrak{h}_L$.

\begin{theorem}\label{thr-3-14}
For each $k\in \mathbb{N}_L$, the operator $\Xi_k$ on $\mathfrak{h}_L$ has a representation on $C(\Gamma_L)$ of the following form
\begin{equation}\label{eq-3-26}
  [\mathsf{F}^{-1}\Xi_k\mathsf{F}f](\sigma)
= \mathbf{1}_{\sigma}(k)f(\sigma\setminus k) + (1-\mathbf{1}_{\sigma}(k))f(\sigma\cup k),\quad \sigma\in \Gamma_L,
\end{equation}
where $f\in C(\Gamma_L)$.
\end{theorem}

\begin{proof}
Since $\{e_{\tau}\mid \tau\in \Gamma_L\}$ is the orthonormal basis for $V(\Gamma_L)$, it suffices to show that for each $\tau\in \Gamma_L$ one has
\begin{equation*}
  [\mathsf{F}^{-1}\Xi_k\mathsf{F}e_{\tau}](\sigma)
  = \mathbf{1}_{\sigma}(k)e_{\tau}(\sigma\setminus k) + (1-\mathbf{1}_{\sigma}(k))e_{\tau}(\sigma\cup k),\quad \sigma\in \Gamma_L.
\end{equation*}
Let $\tau\in \Gamma_L$ be given. Then, with $\mathsf{F}e_{\tau}=Z_{\tau}$, we find
\begin{equation*}
\begin{split}
  \mathsf{F}^{-1}\Xi_k\mathsf{F}e_{\tau}= \mathsf{F}^{-1}\Xi_kZ_{\tau}
  &= \mathsf{F}^{-1}[\mathbf{1}_{\tau}(k)Z_{\tau\setminus k} + (1-\mathbf{1}_{\tau}(k))Z_{\tau\cup k}]\\
  &= \mathbf{1}_{\tau}(k)e_{\tau\setminus k} + (1-\mathbf{1}_{\tau}(k))e_{\tau\cup k},
\end{split}
\end{equation*}
which implies that
\begin{equation*}
  [\mathsf{F}^{-1}\Xi_k\mathsf{F}e_{\tau}](\sigma)
  =  \mathbf{1}_{\sigma}(k)e_{\tau}(\sigma\setminus k)+(1-\mathbf{1}_{\sigma}(k))e_{\tau}(\sigma\cup k),\quad \sigma \in \Gamma_L.
\end{equation*}
This completes the proof.
\end{proof}

The following theorem shows that the operator $\Delta_L$ on $\mathfrak{h}_L$ is unitarily equivalent to the  Laplacian $\widetilde{\Delta}_L$ of the graph
$\Gamma_L$.

\begin{theorem}\label{thr-interpret-Laplace}
The operator $\Delta_L$ on $\mathfrak{h}_L$ has a representation on $C(\Gamma_L)$ in the following manner
\begin{equation}\label{eq}
  \mathsf{F}^{-1}\Delta_L\mathsf{F} = \widetilde{\Delta}_L,
\end{equation}
where $\widetilde{\Delta}_L$ is the Laplacian of the graph $\Gamma_L$ (see Definition~\ref{def-graph-Laplacian} for details).
\end{theorem}

\begin{proof}
Let $f\in C(\Gamma_L)$ be given. By the definition of the operator $\Delta_L$ on $\mathfrak{h}_L$, we have
\begin{equation*}
  \mathsf{F}^{-1}\Delta_L\mathsf{F}
= (L+1)I - \sum_{k=0}^L \mathsf{F}^{-1}\Xi_k\mathsf{F},
\end{equation*}
where $I$ means the identity operator on $C(\Gamma_L)$. Thus, by suing Theorem~\ref{thr-3-14}, we get
\begin{equation*}
\begin{split}
  [\mathsf{F}^{-1}\Delta_L\mathsf{F}f](\sigma)
    &= (L+1)f(\sigma) - \sum_{k=0}^L [\mathsf{F}^{-1}\Xi_k\mathsf{F}f](\sigma)\\
    &= (L+1)f(\sigma) - \sum_{k=0}^L[\mathbf{1}_{\sigma}(k)f(\sigma\setminus k) + (1-\mathbf{1}_{\sigma}(k))f(\sigma\cup k)], \quad \sigma\in \Gamma_L,
\end{split}
\end{equation*}
which together with (\ref{eq-graph-Laplacian-2}) yields
\begin{equation*}
  [\mathsf{F}^{-1}\Delta_L\mathsf{F}f](\sigma)
= [\widetilde{\Delta}_Lf](\sigma),\quad \sigma\in \Gamma_L,
\end{equation*}
which implies that $\mathsf{F}^{-1}\Delta_L\mathsf{F}f=\widetilde{\Delta}_Lf$. It then follows from the arbitrariness of $f\in C(\Gamma_L)$ that
$\mathsf{F}^{-1}\Delta_L\mathsf{F}=\widetilde{\Delta}_L$.
\end{proof}

\begin{remark}\label{rem-3-1}
According to Theorem~\ref{thr-interpret-Laplace}, the operator $\Delta_L$ on $\mathfrak{h}_L$ is unitarily equivalent to the Laplacian $\widetilde{\Delta}_L$ of
the graph $\Gamma_L$. Thus, the Schrodinger equation with $\Delta_L$ as the Hamiltonian is also unitarily equivalent to
that with $\widetilde{\Delta}_L$ as the Hamiltonian.
In other words, our abstract model of CTQW described in Definition~\ref{def-3-1} can be actually viewed as a model of CTQW on the graph $\Gamma_L$.
\end{remark}

\subsection{Perfect state transfer}\label{subsec-3-4}

In the final subsection, we investigate perfect state transfer (PST) in our abstract model of CTQW, namely the CTQW governed by $\Delta_L$.
For brevity, we simply call it the walk $\Delta_L$ below.

Let us first make clear the precise meaning of PST in the walk $\Delta_L$.
For nodes $\sigma$, $\tau\in \Gamma_L$ and $t_0\in \mathbb{R}$,
the walk $\Delta_L$ is said to have PST from $\sigma$ to $\tau$ at time $t=t_0$ if it holds that
\begin{equation}\label{eq-PST}
  \big|\langle \mathrm{e}^{\mathrm{i}t_0\Delta_L} Z_{\sigma}, Z_{\tau}\rangle\big|=1.
\end{equation}

\begin{theorem}\label{thr-3-16}
Let $\sigma$, $\tau \in \Gamma_L$ be given. Then, the following tree conditions are equivalent mutually:
\begin{enumerate}
  \item[(1)]\  $\mathrm{e}^{\mathrm{i}\frac{\pi}{2}\Delta_L} Z_{\sigma} = Z_{\tau}$;
  \item[(2)]\  $\#(\sigma\bigtriangleup \tau)=L+1$;
  \item[(3)]\  $\tau = \mathbb{N}_L\setminus \sigma$.
\end{enumerate}
In that case, the walk $\Delta_L$ has PST from $\sigma$ to $\tau$ at time $t=\frac{\pi}{2}$.
\end{theorem}

\begin{proof}
It is easy to see that under condition (1) the walk $\Delta_L$ has PST from $\sigma$ to $\tau$ at time $t=\frac{\pi}{2}$.
Next, we verify the equivalency of these three conditions. Clearly, condition (2) is equivalent to condition (3). So, we need only to show
that condition (1) and condition (3) are equivalent each other.

In fact, by using diagonal representation (\ref{unitary-diagonal}), we get
\begin{equation*}
\begin{split}
\mathrm{e}^{\mathrm{i}\frac{\pi}{2}\Delta_L} Z_{\sigma}
 & = \sum_{\gamma\in \Gamma_L} \mathrm{e}^{\mathrm{i}(L+1 - \#(\gamma))\pi}\big\langle \widehat{Z}^{(L)}_{\gamma},Z_{\sigma}\big\rangle \widehat{Z}^{(L)}_{\gamma}\\
 & = \frac{1}{\sqrt{2^{L+1}}}\sum_{\gamma\in \Gamma_L} (-1)^{L+1 - \#(\gamma)} (-1)^{\#(\sigma\setminus \tau)}\widehat{Z}^{(L)}_{\gamma}\\
 & = \frac{1}{\sqrt{2^{L+1}}}\sum_{\gamma\in \Gamma_L} (-1)^{L+1 - \#(\gamma)} (-1)^{-\#(\sigma\setminus \tau)}\widehat{Z}^{(L)}_{\gamma}\\
 & = \frac{1}{\sqrt{2^{L+1}}}\sum_{\gamma\in \Gamma_L} (-1)^{L+1 - \#(\sigma\cup\gamma)} \widehat{Z}^{(L)}_{\gamma}.
\end{split}
\end{equation*}
On the other hand, by Theorem~\ref{thr-3-2}, we have
\begin{equation*}
\begin{split}
  Z_{\mathbb{N}_L\setminus \sigma}
   = \sum_{\gamma\in \Gamma_L} \big\langle \widehat{Z}^{(L)}_{\gamma}, Z_{\mathbb{N}_L\setminus \sigma}\big\rangle \widehat{Z}^{(L)}_{\gamma}
   &= \frac{1}{\sqrt{2^{L+1}}}\sum_{\gamma\in \Gamma_L} (-1)^{\#((\mathbb{N}_L\setminus \sigma)\setminus \gamma)} \widehat{Z}^{(L)}_{\gamma}\\
   &= \frac{1}{\sqrt{2^{L+1}}}\sum_{\gamma\in \Gamma_L} (-1)^{\#((\mathbb{N}_L\setminus (\sigma\cup \gamma))} \widehat{Z}^{(L)}_{\gamma}\\
   & = \frac{1}{\sqrt{2^{L+1}}}\sum_{\gamma\in \Gamma_L} (-1)^{L+1 -\#(\sigma\cup \gamma)} \widehat{Z}^{(L)}_{\gamma}.
\end{split}
\end{equation*}
Thus, we come to a useful equality of the form $\mathrm{e}^{\mathrm{i}\frac{\pi}{2}\Delta_L} Z_{\sigma} =Z_{\mathbb{N}_L\setminus \sigma}$.
With this equality in mind, we can immediately see the equivalence of condition (1) and condition (3).
\end{proof}

Recall that $\mathbb{N}_L=\{0,1,\cdots,L+1\}$ and $\Gamma_L$ is the power set of $\mathbb{N}_L$.
Thus, $\mathbb{N}_L\setminus \sigma \in\Gamma_L$ whenever $\sigma\in\Gamma_L$.
This together with the above proof justifies the next corollary.

\begin{corollary}\label{coro-3-17}
For all $\sigma\in \Gamma_L$, it holds true that\, $\mathrm{e}^{\mathrm{i}\frac{\pi}{2}\Delta_L} Z_{\sigma} =Z_{\mathbb{N}_L\setminus \sigma}$,
in particular the walk $\Delta_L$ has PST from node $\sigma$ to node $\mathbb{N}_L\setminus \sigma$
at time $t=\frac{\pi}{2}$.
\end{corollary}

\section{Conclusion remarks}\label{sec-4}

As indicated above, we call the function $\overline{P}(\cdot)$ defined by (\ref{eq-time-average distribution}) the time-average probability distribution of the walk $\Delta_L$.
In fact, a justification for the name lies in the following observations.
Just consider the function $t\mapsto P_t(\sigma)$ on $\Gamma_L$, where $P_t(\sigma)$ is the probability that the walker is found at node $\sigma\in $ at time $t$.
According to Theorem~\ref{thr-periodicity}, for each $\sigma \in \Gamma_L$, $t\mapsto P_t(\sigma)$ is a bounded periodic function with $\pi$ being a period,
which together with a limit procedure yields
\begin{equation*}
  \lim_{T\to +\infty}\frac{1}{2T}\int_{-T}^TP_t(\sigma)dt = \frac{1}{\pi}\int_0^{\pi}P_t(\sigma)dt.
\end{equation*}
Thus $\overline{P}(\cdot)$ can be redefined as
\begin{equation}\label{eq}
  \overline{P}(\sigma) =  \lim_{T\to +\infty}\frac{1}{2T}\int_{-T}^TP_t(\sigma)dt,\quad \sigma \in \Gamma_L,
\end{equation}
and $\overline{P}(\sigma)$ is exactly the average probability of finding the walker at node $\sigma$ over the whole time interval $\mathbb{R}=(-\infty,\infty)$.


\begin{thebibliography}{99}

\bibitem{alvir} B. Alvir, S. Dever, B. Lovitz, J. Myer, C. Tamon, Y. Xu, and  H. Zhan,  Perfect state transfer in Laplacian quantum walk,
   J. Algebraic Combin. 43 (2016), no. 4, 801-826.

\bibitem{balu}  R. Balu, Quantum walks on regular graphs with realizations in a system of anyons, Quantum Inf. Process. 21 (2022), Article number 177.

\bibitem{bose} S. Bose, Quantum communication through an unmodulated spin chain,  Phys. Rev. Lett. 91 (2003), no. 20, 207901.

\bibitem{chen-1} J.S Chen, Invariant states for a quantum Markov semigroup constructed from quantum Bernoulli noises, Open Syst. Inf. Dyn. 25 (2018), no. 4, 1850019.

\bibitem{chen-2} J.S Chen, Quantum Feller semigroup in terms of quantum Bernoulli noises, Stoch. Dyn. 21 (2021), no. 4, 2150015.

\bibitem{cohn} D.L. Cohn, Measure Theory, 2nd Edition, Springer, New York (2013).

\bibitem{cheung-godsil} W.-C. Cheung and C. Godsil, Perfect state transfer in cubelike graphs, Linear Algebra Appl. 435 (2011), no. 10, 2468-2474.

\bibitem{childs} A. M. Childs, On the relationship between continuous- and discrete-time quantum walk, Commun. Math. Phys. 294 (2010), 581-603.

\bibitem{christandl} M. Christandl,  N. Datta, A. Ekert and A. Landahl, Perfect state transfer in quantum spin networks, Phys. Rev. Lett. 92 (2004), 187902.

\bibitem{farhi-gutmann}  E. Farhi and S. Gutmann, Quantum computation and decision trees, Phys. Rev. A 58 (1998), 915-928.

\bibitem{godsil} C. Godsil, State transfer on graphs, Discrete Math. 312 (2012), no. 1, 129-147.

\bibitem{kempton} M. Kempton, G. Lippner and S-T. Yau, Perfect state transfer on graphs with a potential, Quantum Inf. Comput. 17 (2017), no. 3-4, 303-327.

\bibitem{kendon} V. Kendon and C. Tamon, Perfect state transfer in quantum walks on graphs, J. Comput. Theoret. Nanosci. 8 (2011), no. 3, 422-433.

\bibitem{konno} N. Konno, Limit theorem for continuous-time quantum walk on the line, Phys. Rev. E, 72 (2005), no. 2, 026113.

\bibitem{lin} Y. Lin, G. Lippner and S-T. Yau, Quantum tunneling on graphs, Comm. Math. Phys. 311 (2012), no. 1, 113-132.

\bibitem{obata} N. Obata, Spectral Analysis of Growing Graphs: A Quantum Probability Point of View, Springer, Singapore (2017).

\bibitem{wangcs-1} C.S. Wang and J.S. Chen, Quantum Markov semigroups constructed from quantum Bernoulli noises, J. Math. Phys. 57 (2016), no. 2, 023502.

\bibitem{wangcs-2} C.S. Wang, Y.L. Tang and S.L. Ren, Weighted number operators on Bernoulli functionals and quantum exclusion semigroups,
 J. Math. Phys. 60 (2019), no. 11, 113506.

\bibitem{wangcs-4} C.S. Wang, C. Wang, S.L. Ren and Y.L. Tang, Open quantum random walk in terms of quantum Bernoulli noise, Quantum Inf. Process.
17 (2018), no. 3, Paper No. 46.

\bibitem{wangcs-3} C.S. Wang and X.J. Ye, Quantum walk in terms of quantum Bernoulli noises, Quantum Inf. Process. 15 (2016), no. 5, 1897-1908.

\bibitem{venegas} S. E. Venegas-Andraca, Quantum walks: a comprehensive review, Quantum Inf. Process. 11 (2012), 1015-1106.

\end{thebibliography}
\end{document}